\documentclass[10pt,nocopyrightspace]{sigplanconf} 

\usepackage{balance}
\usepackage[latin1]{inputenc}
\usepackage{amsmath}
\usepackage{amssymb}
\usepackage{stmaryrd} 
\usepackage{array}
\usepackage{url}
\usepackage{verbatim}
\usepackage{xspace}
\usepackage{mathpartir}
\usepackage{color}
\usepackage{fppdf}

\definecolor{Cquestioncolor}{rgb}{.8,.0,.0}

\definecolor {colorgrey}{rgb}{.75,.75,.75} 

\newcommand{\Sc}[1] {\, #1 \,}

\newcommand{\E}[1] {{\em #1}}       
\newcommand{\F}[1] {\textsf{#1}}    

\newcommand{\sref}[1]{\S\ref{#1}}

\newcommand{\cref}[1]{Claim~\ref{#1}}

\usepackage{ntheorem}

\newtheorem{fact}{Fact}[section]

\newenvironment{proof}[1][Proof]{\begin{trivlist}
\item[\hskip \labelsep {\bfseries #1}]}{\end{trivlist}}

\newcommand{\out}[1] {}

\newcounter{codeLineCntr}

\reversemarginpar
\newif\ifnotes
\notestrue

\newcommand{\punt}[1]{}

\newcommand{\secref}[1]{Section~\ref{sec:#1}}

\newcommand{\figref}[1]{Figure~\ref{fig:#1}}

\newcommand{\fctref}[1]{Fact~\ref{fct:#1}}

\newcommand{\proc}[1]{\ifmmode\mbox{\textsc{#1}}\else\textsc{#1}\fi}

  \newcommand{\func}[1]{\ifmmode\mathrm{#1}\else\textrm{#1}fi} %

\newcounter{remark}[section]
\newcommand{\myremark}[3]{
\refstepcounter{remark}
~\\
{\bf $[$ \scriptsize{#1}.{\theremark}:}{\textbf{\scriptsize{#3}}$]$}
\\
}

\renewcommand{\myremark}[3]{{\bf TODO[#1]: #3}}
\renewcommand{\myremark}[3]{}

\newcommand{\uremark}[1]{\myremark{Umut}{U}{#1}}
\newcommand{\aremark}[1]{\myremark{Arthur}{U}{#1}}

\newcommand{\ur}[1]{\uremark{#1}}

\renewcommand{\topfraction}{0.999999}

\renewcommand{\floatpagefraction}{0.9999999}

\usepackage{amsmath,amsfonts,amssymb,graphicx}

\title{Analysing and Understanding Parallel Speedups}
\title{Techniques for Empirical Analysis of Memory Effects in Parallel
  Programs}

\title
{Parallel Work Inflation, Memory Effects, and their Empirical Analysis
}

\authorinfo{Umut A. Acar}
          {Department of Computer Science \\  
            Carnegie Mellon University \& Inria}
          {umut@cs.cmu.edu}

\authorinfo{Arthur Chargu\'eraud}
          {Inria}
          {arthur.chargueraud@inria.fr}

\authorinfo{Mike Rainey}
          {Inria}
          {mike.rainey@inria.fr}

\renewcommand{\topfraction}{0.999999}

\renewcommand{\floatpagefraction}{0.8}

\newcommand{\smallheight}{2.2in}
\newcommand{\smallwidth}{2.2in}

\newcommand{\infw}[2]{F({#1},{#2})}
\newcommand{\infwc}[1]{F({#1})}

\newcommand{\figanalysisadjustheight}{\vspace{0.2em}}

\DeclareMathOperator{\Exp}{\mathbb{E}}

\begin{document}

\CopyrightYear{2017}
\copyrightdata{}

\maketitle


\begin{abstract}
In this paper, we propose an empirical method for evaluating the
performance of parallel code. Our method is based on a simple idea
that is surprisingly effective in helping to identify causes of poor
performance, such as high parallelization overheads, lack of adequate
parallelism, and memory effects. Our method relies on only the
measurement of the run time of a baseline sequential program, the run
time of the parallel program, the single-processor run time of the
parallel program, and the total amount of time processors spend idle,
waiting for work.

In our proposed approach, we establish an equality between the
observed parallel speedups and three terms that we call parallel work,
idle time, and work-inflation, where all terms except work inflation
can be measured empirically, with precision.  We then use the equality
to calculate the difficult-to-measure work-inflation term, which
includes increased communication costs and memory effects due to
parallel execution. By isolating the main factors of poor performance,
our method enables the programmer to assign blame to certain
properties of the code, such as parallel grain size, amount of
parallelism, and memory usage.

We present a mathematical model, inspired by the work-span model, that
enables us to justify the interpretation of our measurements. We
also introduce a method to help the programmer to visualize both the
relative impact of the various causes of poor performance and the
scaling trends in the causes of poor performance. Our method fits in a
sweet spot in between state-of-the-art profiling and visualization
tools. We illustrate our method by several empirical studies and
we describe a few experiments that emphasize the care that
is required to accurately interpret speedup plots.
\end{abstract}

\section{Introduction}

In the current state of the art, implementing a parallel algorithm on
a multicore machine requires more than translating the algorithm to a
parallel program by using a language or a parallelism API, such as
OpenMP~\cite{openmp}, TBB~\cite{ThreadingBuildingBlocksManual},
X10~\cite{x10-2005}, or Cilk Plus~\cite{IntelCilkPlus}.  During the
development cycle, the programmer will likely have to {\em tune} their
implementation by experimenting with several important parameters and
optimizations in order to elicit decent performance.  To this end, the
programmer typically compares the performance of the parallel code
with multiple processors to the performance of a sequential baseline
and computes the {\em speedup} achieved as the ratio of the time for
the baseline to the time for the multiprocessor run. It is well known
that, for this comparison to be meaningful, the baseline has to be
selected carefully and must be an optimal sequential algorithm and an
optimized implementation.

After an initial implementation, speedup curves that the programmer
obtains usually resemble those that are shown in
\figref{poor-speedups}.  Three of the speedup curves are taken from runs of
three different configurations of the Cilksort benchmark
~\cite{FrigoLeRa98}, and the other speedup is taken from one run of
the Maximal Independent Set
benchmark~\cite{Blelloch:2012:IDP:2145816.2145840}. These speedups
scale poorly, deviating significantly from the linear optimum. Faced
with such results, the programmer has to study the performance of the
code to identify and eliminate causes of suboptimal performance.

There are four main non-overlapping {\em factors} that contribute to
suboptimal parallel performance.
\begin{itemize}
\item \emph{Algorithmic overheads}, which correspond to
  the difference in the amount
  of work performed by the sequential baseline program and the
  sequential execution of the parallel program.

\item \emph{Scheduling overheads}, which consists of the cost of creating
  threads plus the cost of performing load balancing.

\item \emph{Lack of parallelism} in the application, that leads to {\em
  idling} processors which are starving for work.

\item \emph{Work inflation}, which we define as the increase in the
  cost of the operations performed in a parallel run compared with a
  single-processor run, when executing the parallel code.
\end{itemize}

Note that the first and the last factors are different.  
On the one hand, algorithmic overheads result primarily from the fact that a
parallel algorithm is usually more complex than a sequential algorithm
for the same problem. 
On the other hand, work inflation
measures the increase in the work of the parallel implementation as we
increase the number of processors.  Work inflation includes memory
subsystem effects, and the costs for communication, synchronization
such as memory fences and atomic operations, false sharing,
maintenance of cache coherency, contention at the memory bus, and
memory consistency protocol. Because work inflation occurs at the
hardware level, the overall impact of work inflation is difficult, if
not impossible, to measure directly. 

A key step in the tuning process is that of identifying which of
the four factors are significant. 
For example, as we will see, each of
the speedup curves in \figref{poor-speedups} is poor due to just one
or two of the four factors. By just looking at the speedup curves, it
is not possible to determine which factors harm scalability and by how
much.  In general, despite their ability to show scaling trends,
speedup curves can, by themselves, provide only vague hints
into what factors harm scalability.

Although there are several performance tools to analyze parallel
applications, there are currently neither tools nor widely-known
methods that enable programmers to analyze the relative impact on
scalability of the different factors, such as those listed above. The
Cilkview analyzer can be used to predict the scalability of an
application based on the logical parallelism expressed in control
structure of the code~\cite{HeLeiserson10}. However, if the code
expresses plenty of parallelism, Cilkview analyzer is unlikely to
provide additional insights into the causes of poor performance. Tools
such as Intel Thread Profiler~\cite{IntelParallelAdvisor}, Intel
Parallel Amplifier~\cite{IntelParallelAmplifier},
HPCToolkit~\cite{TallentMellorCrummey07},
Kismet~\cite{JeonGarciaLouie11} and
Kremlim~\cite{GarciaJeonDonghwan11} can provide detailed information
on the utilization of the processors over time and on the breakdown of
the relative importance of the subroutines of the
program. Predator~\cite{LiuBergerPredator} can detect false sharing in
instrumented runs of application code. Each of these tools fills an
important gap in the toolkit of a parallel programmer. Nevertheless,
none of these tools are suitable for analyzing all types of performance
issues.

A separate issue relates to profiling instrumentation.
Cilkview relies on binary instrumentation and analyzes only
instruction counts.  The other tools rely on various other forms of
instrumentation. Such instrumentation increases the risk
that instrumentation-specific overheads will 
themselves influence performance, and the overheads will do so
in ways that obscure the performance issues of interest.
Although it is sometimes essential to understand certain aspects of
performance, heavyweight instrumentation causes interference that can
obscure the global picture, that is, the performance of the production
binary, which typically has little or no instrumentation.

\begin{figure}
\begin{center}
\includegraphics[width=\smallwidth,height=\smallheight]{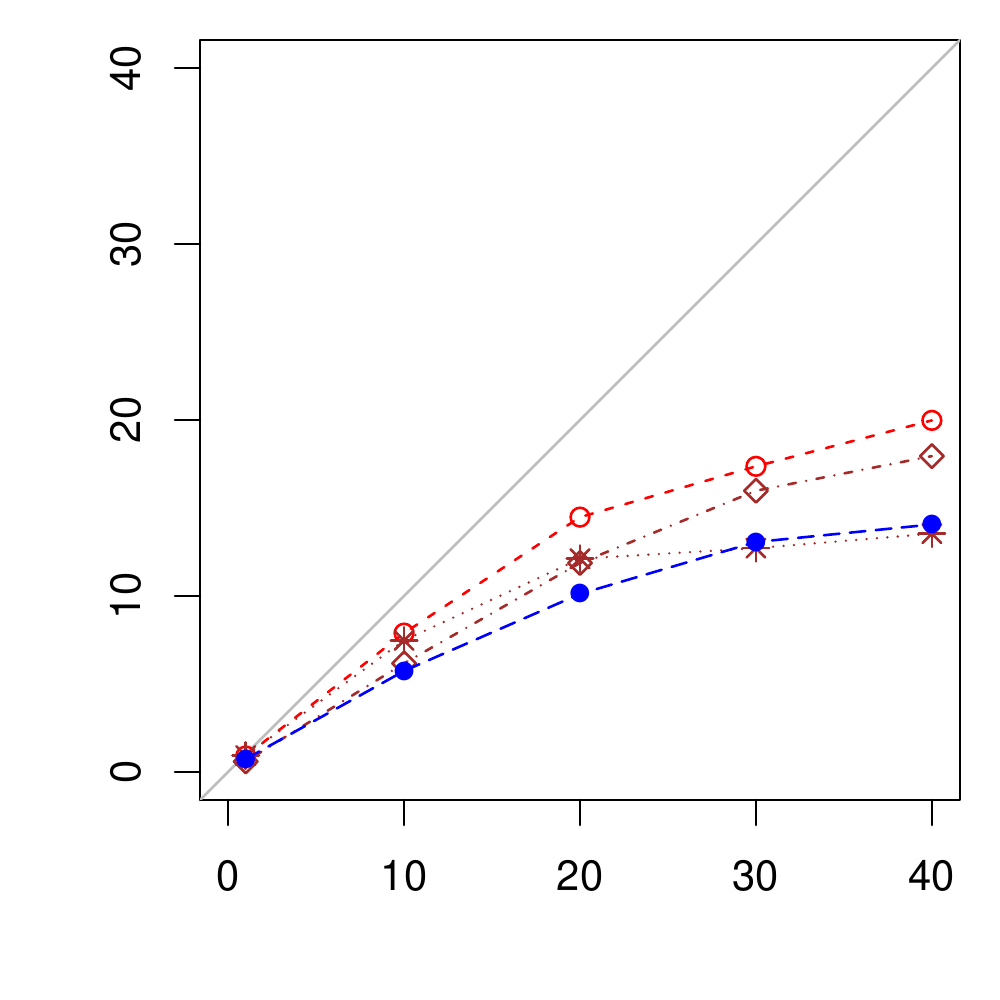}
\end{center}
\vspace{-2.0em}
\caption{Traditional speedup curve showing four poor speedups. The y axis
represents the speedup and the x axis the number of processors.}
\label{fig:poor-speedups}
\end{figure}

In this paper, we present an experimental method for diagnosing
observed performance and scalability problems.  Our proposal rests on
some simple observations but it provides a surprisingly effictive and
non-intrusive approach.

Our approach relies on the following measurements: 
\begin{itemize}
\item the sequential execution time of the baseline program,
\item the sequential execution time of the parallel program (i.e., the
  running time of the parallel program using a single processor),
\item the parallel execution time of the parallel program (with
  different numbers of cores),
\item the total time that processors spend idling (waiting for work).
\end{itemize}


Using these measures, we show that it is possible to derive the amount
of work inflation, a quantity that is difficult to measure directly.
More generally, we are able to calculate the amount of speedups lost
due to overheads associated with the parallel algorithm, the amount of
speedups lost due to idle time, and the amount of speedups lost due to
the work inflation. By measuring and calculating these values for
various number of processors, we can study scalability trends, and the
factors contributing the observed results.  As we describe 
(\sref{sec:measures}), these
quantities can be measured unintrusively, without heavy
instrumentation of the binary, and are therefore representative of the
actual, observed performance (they are not based on simulations or
profiling information). 



Using such measurements, we propose an approach to visulazing
important performance information in the form of {\em factored speedup
  plots} that include three additional speedup curves, all of which are
calculated with respect to the optimized sequential baseline.  These
plots enable studying the different contributing factors to the
speedups.

\begin{itemize}
\item A {\em maximal speedup plot} shows the speedups that the program
  would obtain  if we ignore work inflation and idle time. In other
  words, the maximal speedup plot shows the speedups that would be
  achieved if the speedup of the parallel program were scaling up
  linearly with the number of processors.

\item An {\em idle-time-specific speedup plot} takes into account idle
  time but ignores work inflation.  In other words, idle-time-specific
  speedups represent the speedups that would be obtained if only the
  idle time and algorithmic overheads (the overheads of the parallel
  program with respect to the baseline) were preventing the program
  from achieving maximal speedups.

\item An {\em inflation-specific speedup plot} shows the speedups that
  the program would obtain if we ignore idle time. In other words, the
  inflation-specific speedup represents the speedup that would be
  achieved if only the work inflation and the algorithmic overheads
  were preventing the program from achieving maximal speedups.
\end{itemize}

\figref{factored-speedup-poor} shows the factored speedup curves for
the Cilksort benchmark with one specific configuration.  Our factored
speedup plot enables the programmer to visualize all three curves at
once. The plot conveys three types of information: (1) the absolute
position of the curves, (2) the relative position of the curves (i.e.,
the gaps between the curves), (3) the shapes of the curves (i.e.,
curvature), which informs on the scaling of specific speedup curves.

\paragraph{Workflow.}
Our factored speedup plot plays a complementary role to the
parallel-performance analyzers. If the factored speedup plot suggests
lack of parallelism in the application, then the programmer may choose
to find the bottleneck using a tool such as Cilkview.  If the program
is large and it is not clear which pieces of the code to blame for
lack of parallelism, the programmer may search for the most
significant regions of code using one of the tools, such as HPCToolkit
and Kremlin. If work inflation is high, the programmer may choose to
look for potential false sharing with Predator, for
example. After a problematic region of code is identified, the
programmer may synthesize from the region a smaller benchmark program
and repeat the process from above.

Unlike many other other methods for analyzing performance of parallel
codes, ours enables the programmer to observe their production code
directly. As such, our method fits into a gap that we believe exists
between traditional speedup plots and existing parallel performance
tools, such as Cilkview. The close correspondence between the
production binary and our lightweight profiling binary is possible
thanks to the property that instrumentation we insert into the program has
no noticable impact on the performance of the parallel code.

\begin{figure}
\begin{center}
\includegraphics[width=3.0in,height=3.0in]{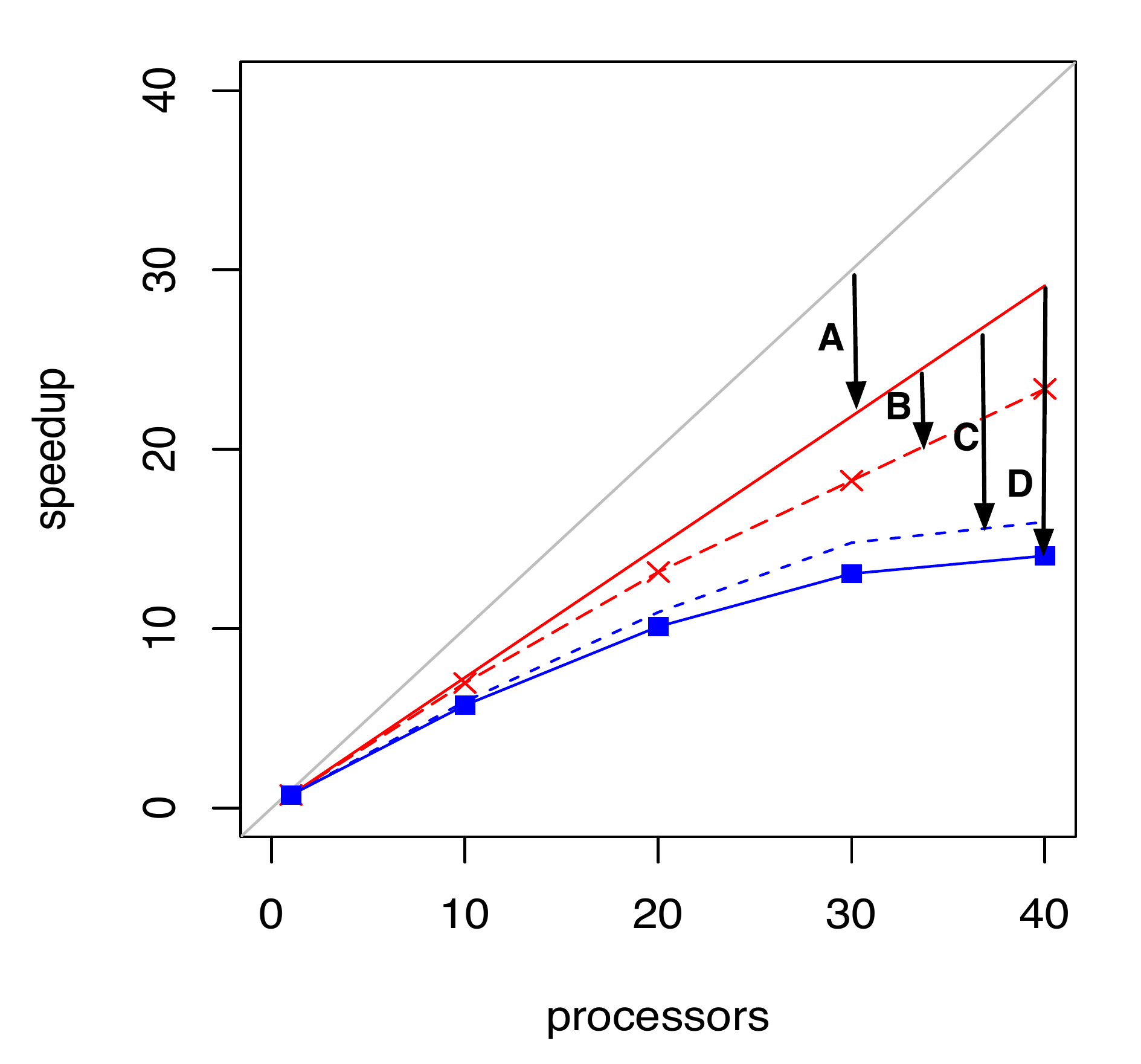}
\end{center}
\vspace{-1.0em}
\caption{Factored speedup curves (top to bottom): linear, maximal,
  idle-time specific, inflation specific, and actual speedups for
  Cilksort. The arrows indicate gaps between the speedup curves that
  helps identify the contribution of each factor.
  }
\label{fig:factored-speedup-poor}
\end{figure}

While developing algorithms and studying their efficiency with the
help of factored speedup plots, we have often been impressed by how
much work inflation (and thus speedups) could be affected, in
counterintuitive ways, by the degree of optimization of the program
code, and by the size of input data with respect to the size of the
cache. In order to illustrate the extent to which speedups can be
affected by these two aspects, we complete our paper with
microbenchmark studies demonstrating how seemingly minor changes in
the parameters can significantly affect the speedups measured.

Our contributions are as follows.
\begin{itemize}
\item We present a model, inspired by the work-span model, which 
  accounts for work inflation, even though work inflation cannot
  be predicted by any theory and cannot be measured directly.
\item We introduce {\em factored speedup plots} as a practical technique for
  visualizing the relative contribution of each of the three main sources of 
  slowdown considered by our model: overheads, idle time and work inflation.
\item We describe two artifacts that may significantly affect the
  interpretation of speedup curves.
  Although the
  existence of these effects is well know, we believe that the degree
  to which they can impact work inflation is often underappreciated.
\end{itemize}

\section{\hspace{-8pt}A Method for Diagnosing Performance Problems}%

We describe a method for diagnosing problems with performance and
scalability by identfying the contributions of the factors that
mentioned in the introduction.  For the purposes of mathematical
simplicity, at first, we do not consider scheduling overheads, and
we assume that our measurements (programs and schedulers)
are deterministic.  We later describe how to account for
non-determinism (\secref{method::non-deterministic}) and scheduling
overheads (\secref{method::scheduling}).

\subsection{Measures}
\label{sec:measures}

Given a parallel program, and given an associated sequential baseline program,
our approach relies on the four following measures.
\begin{itemize}
\item $T_s$, the execution time of the sequential baseline.

\item $T_P$, the execution time of the parallel program with $P$ cores.

\item $I_P$, the total idle time associated with the parallel program
(measured by instrumenting the scheduler).


\item $T_1$, the 1-core execution time of the parallel program. We
  call $T_1$ the ``parallel work with 1 core''.
\end{itemize} 

Measuring $T_s$, $T_P$ and $T_1$ is achieved by querying the
system time at the beginning and at the end of the executions.
In particular, it does not require any instrumentation of
the code being benchmarked.
Measuring $I_P$ is just slightly more complex. We measure the total 
idle time by instrumenting the main loop of the scheduler code that is
executed by each core, and which handles load balancing operations. 
We compute for each core
the sum of the duration of the periods of time during
which the core is waiting to acquire work. We call such periods
\E{idle phases}, and we measure their duration using 
unobtrusive cycle-counter instructions, which are provided
by modern multicore machines.

The total cost of our instrumentation of the scheduler is
negligible in front of the execution time of the program.
For each idle phase, we perform two queries to the cycle counters,
and update one field from the thread-local storage.
To end an idle phase, the processor needs to receive at least one task,
and the time required to complete the execution of this task is in general
a lot greater than the cost of measuring the duration of the idle phase.

Moreover, when a work-stealing scheduler is used, the total number of
idle phases is relatively small. More precisely,
the number of idle phases is bounded by $P-1$ plus
the number of steals, because initially all cores are idle but one,
and each idle phase can only end as a result of a successful steal.
Analysis of work stealing shows that, for all programs that exhibit
sufficient parallelism, the number of steals is, with high probability;
relatively small in front of the total number of tasks~\cite{acarchra13}.
In summary, the overhead of our instrumentation is, for all practical
purposes, negligible in front of the total execution time.

\subsection{Definitions}
\label{sec:defs}

Using the four measurements stated above,
we derive two additional quantities.
\begin{itemize}
\item $W_P$, the parallel work with $P$ cores.

\item $F_P$, the work inflation with $P$ cores.
\end{itemize} 

To see how to calculate these additional quantities, we start with a
simple fact.
\begin{fact}[time decomposition]
The total amount of time available to the $P$ cores during a run that
lasts $T_P$ time decomposes in work time and idle time.
\[
P \cdot T_P \Sc{=}  W_P + I_P. 
\]
\end{fact}
This fact makes it immediately possible to calculate $W_P$.

Recall that we define work inflation to be the increase in work as a
result of parallel execution.  This leads to the following fact.
\begin{fact}[definition of work inflation]
Work inflation (at $P$ cores) is the difference between the work performed
by the parallel program when using $P$ cores and the work performed by the same program
when using a single core. We therefore have:
$$F_P \Sc{=} W_P - T_1.$$
\end{fact}
As shown by the fact below, we can calculate the work inflation $F_P$
%
\begin{fact}[formula for work inflation]
\label{fct:formula-work-inflation}
$$F_P \Sc{=} P \cdot T_P - I_P - T_1.$$
\end{fact}

For the purpose of analysing speedups (\secref{factored}) 
and of comparison with the work-span model (\secref{compare}),
we combine the previous facts so as to obtain a reformulation
of the parallel execution time in terms of the values
of $T_1$ (1-core execution time of the parallel program), 
$I_P$ (idle time) and $F_P$ (work inflation).
\begin{fact}[reformulation of parallel time]
\label{fct:reformulation}
The parallel execution time can be expressed as follows:
$$T_P \Sc{=} \frac{T_1 + I_P + F_P}{P}.$$
\end{fact}

\subsection{Factored speedup plots}

\label{sec:factored}

In order to better understand the effect of work inflation and idle time
on the speedup values achieved by parallel programs, we reformulate,
using \fctref{reformulation}, the expression of speedup values,
which is defined as the baseline time divided by the parallel time.

\begin{fact}[reformulation of speedups]
The speedup at $P$ cores can be reformulated as follows:
$$\F{speedup} \Sc{=} \frac{T_s}{T_P} \Sc{=} \frac{P \cdot T_s}{ T_1 + I_P + F_P}$$
\end{fact}

Starting with this formula, we propose four speedup measures that
offer upper bounds of varying degrees of precision.  Analyzing these
speedups and the gaps between them we can determine the effects of
work inflation and other characteristics of the computation on the
performance.

\paragraph{Linear speedups.}
When using $P$ cores to perform a computation, we generally do not
expect the parallel execution to be more than $P$ times faster than
the sequential baseline. Therefore, the \E{linear speedup} at $P$ cores
is equal to the value $P$.

\paragraph{Maximal speedups.}
We define the quantity
\[
\frac{P \cdot T_s}{T_1}
\]
as the maximal speedup because it assumes the work inflation and the
idle time to be zero. Maximal speedups offer a realistic
upper bound on the parallel speedup by taking into account the
(possibly) additional work that must be performed by the parallel run
in relation to the sequential run.  

\paragraph{Idle-time-specific speedup.}
We define the quantity
\[
\frac{P \cdot T_s}{T_1 + I_P}
\]
as the idle-time-specific speedup because it assumes work inflation to
be zero but takes into account available parallelism (as measured by
the idle time).  

\paragraph{Inflation-specific speedup.}
We define the quantity
\[
\frac{P \cdot T_s}{T_1 + F_P}
\]
as the inflation-specific speedup because it assumes idle time to be
zero but takes into account work inflation. In the formula above,
since we cannot measure $F_P$ directly, we deduce it from $T_P$ and $I_P$.
More precisely, the inflation-specific speedup is computed as
${(P \cdot T_s)}/{(P \cdot T_P - I_P)}$.

\paragraph{Actual Speedups.}
By definition, the actual speedup is: 
$$\frac{T_s}{T_P}$$

\subsection{Minding the gap}

The three forms of speedups help analyze the empirical behavior of a
parallel algorithm by isolating several different effects into
different curves.  \figref{factored-speedup-poor} shows an example.
The linear speedup is drawn as the diagonal.  Right below it is the
maximal speedup, drawn with a solid black line. 

The gap labelled
\textbf{A} between the linear and the maximal speedups shows the
amount of the algorithmic overheads that can be expected from
parallelization and the overhead of thread creation.  In other words,
we can expect to match maximal speedups if the computation is fully
parallel and only on parallel hardware that is able to support all
operations with excellent scalability. 

Right below the maximal speedup
curve lies the idle-time-specific speedup curve, which takes into
account the amount of parallelism but excludes work inflation. The
gap labelled \textbf{B} between idle-time-specific speedup and the
maximal speedup shows the idle time, which, assuming an close-to-greedy
scheduler and a sufficiently-fine granularity of the tasks,
reflects the scarcity of parallelism in the computation:
the larger the gap, the scarcer is parallelism.  We can expect to
match idle-time-specific speedups only on parallel hardware that
exhibits no noticable communication overheads and is able to scale
memory operations well.  

Right below the idle-time-specific speedup
curve lies the inflation-specific speedup curve.  
The gap labelled
\textbf{C} between the maximal speedup and the inflation-specific
speedup illustrates the amount of work inflation: the larger the gap,
the greater the work inflation.  

At the bottom, the actual speedup curve reports on the speedups
actually measured. The speedups include all measured factors
(algorithmic overheads, idle-time, and work inflation.
The gap labelled \textbf{D} illustrates the amount of speedup 
lost to idle time and work inflation combined.
Finally, note that the gap between the actual speedup curve
and idle-time-specific speedup curve indicates the amount 
of work inflation, and that the gap between the actual speedup 
curve and the inflation-specific curve speedup indicates the
amount of idle time.

\vspace{-4pt}
\subsection{Minding the curvature}

In addition to studying the space between the curves, 
it is often also possible to deduce useful information 
from the curvature of the curves. A few features
are particularly informative.

If the maximal speedup curve is not a straight line but instead tends
to flatten, then this curve indicates that the amount of overhead
increases with the number of cores.  In such case, the algorithm
presumably would not scale up well with the number of cores.

Let us assume that the overhead curve appears as a straight line.  If
the idle-time-specific curve flattens towards a horizontal line, then
this curve indicates that the additional computation time provided by
using more cores is mostly wasted as idle time.  Presumably, the
program lacks parallelism.

If the inflation-specific curve flattens towards a horizontal line,
then this curve indicates that the additional computation time
provided by using more cores is almost entirely converted into work
inflation.  This situation is generally characteristic of a memory
bottleneck that limits the throughtput of the operations performed on
the main memory.

If the inflation-specific curve ends up slopping downwards, then this
curve indicates that using more cores actually degrades the
performance of the parallel program.  This situation is typically
caused by synchronization, in particular extensive use of either
atomic operations or false sharing or both.

\vspace{-4pt}
\subsection{Work-span model versus inflation model}
\label{sec:compare}

Comparing our proposed model with the work-span model brings out
interesting similarities and differences between the two approaches.
Consider a parallel program with work $T_1$ and span $T_{\infty}$, and
whose parallel time is $T_P$ on $P$ processors.  In the work span
model (based on Brent's theorem~\cite{Brent74}, using a greedy
scheduler), the parallel time, ignoring scheduling costs, is bounded
as 
$$T_P \Sc{\leq} \frac{T_1}{P} + T_{\infty}.$$

By comparison, our model does not provide an upper bound, but instead
the following  exact equality that involves (measured) idle time and (derivable)
work inflation (\fctref{reformulation}),
$$T_P \Sc{=} \frac{T_1}{P} + \frac{I_P}{P} + \frac{F_P}{P}. $$

If we ignore scheduling costs, then there are two important differences
between the work-span model and our approach.  First, our proposal
relies on the measurement of the actual average amount of idle time
per core (that is, ${I_P}/{P}$), rather than an upper
bound computed as a property of the computation (that is, the span
$T_{\infty}$).  Second, our proposal includes a term for work
inflation, whereas the work-span model does not account
for differences between the uniprocessor work
and multiprocessor work.


\subsection{Generalization to non-deterministic executions}
\label{sec:theory}
\label{sec:inflation}
\label{sec:method::non-deterministic}

In this section, we justify that our approach naturally extends 
to non-deterministic executions.
We call \E{run} a particular instance of a program execution. In particular,
for a parallel execution, the run describes the \E{schedule}, that is,
for each instruction, the time at which and the core on which it gets executed.
We let $\mathcal{R}_s$ denote the set of runs of the sequential baseline,
$\mathcal{R}_P$ the set of runs of the parallel program with $P$ cores,
and $\mathcal{R}_1$ the set of runs of the parallel program with $1$ core.
We let $R$ be a random variable ranging over one of these three sets of runs.

\begin{itemize}
\item Given a run $R$ in $\mathcal{R}_s$, we let $T_s(R)$ denote
the execution time of this run.

\item Given a run $R$ in $\mathcal{R}_P$, we let $T_P(R)$ denote 
the execution time of this run.

\item Given a run $R$ in $\mathcal{R}_P$, we let $I_P(R)$ denote 
the total idle time involved in this run.

\item Given a run $R$ in $\mathcal{R}_1$, we let $T_1(R)$ denote
the execution time of this run.
\end{itemize} 

We then define the parallel work of a run as follows:
$$ W_P(R) \Sc\equiv P \cdot T_P(R) - I_P(R)$$
We define the work inflation of a run as the difference between
the parallel work of this run and the \E{expected} work of a 1-core run.
Regarding the latter, we write  $\Exp[T_1(R')]$ the expected execution
time of a random run $R'$ in $\mathcal{R}_1$.
The formal definition of work inflation is thus:
$$F_P(R) \Sc\equiv W_P(R) - \Exp[T_1(R')].$$

We write $T_s$ the expected value of $T_s(R)$, for $R$ in $\mathcal{R}_s$.
Similary, we write $T_P$ and $I_P$ the expected values of $T_P(R)$ and $I_P(R)$,
respectively, for $R$ in $\mathcal{R}_P$, and
write $T_1$ the expected value of $T_1(R)$ for $R$ in $\mathcal{R}_1$.
(Note that $T_1$ is the same as $\Exp[T_1(R)]$.)
With this notation, the earlier definitions given for the deterministic case 
can be applied without any modification. In particular, we define:

\begin{itemize}
\item Ideal speedup, as the value $P$.

\item Maximal speedup, as the value $\frac{P \cdot T_s}{T_1}$.

\item Idle-time-specific speedup, as the value $\frac{P \cdot T_s}{T_1 + I_P}$.

\item Inflation-specific speedup, as the value $\frac{P \cdot T_s}{P \cdot T_P - I_P}$.

\item Actual speedup, as the value $\frac{T_s}{T_P}$.
\end{itemize} 

Observe that, in the formulae above, we have chosen to compute the ratios
of expected values, as opposed to the expected values of ratios.
For example, we define actual speedup as
$\frac{\Exp[T_s(R)]}{\Exp[T_P(R')]}$
and not as $\Exp[\frac{T_s(R)}{T_P(R')}]$.
The alternative choice would also be possible.
However, we believe that, given a sample of measured runs,
it makes more sense to report the speedup associated with 
the average execution time, rather than to report the average
speedup value, because speedups are only a tool for analysing
performance, whereas the execution time is what we ultimately
care to minimize.

\subsection{Accounting for scheduling costs}

\label{sec:method::scheduling}

In this section, we explain how our approach smoothly generalizes
to take scheduling costs into account.
We start by introducing the following additional variables:
\begin{itemize}
\item $S_1$, the scheduling work of a 1-core run of the parallel program.
\item $S_P$, the scheduling work of a P-core run of the parallel program.
\item $W_1$, the user work of a 1-core the parallel program. We call
  ``user work'' the work performed by the user code as opposed
  to that performed by the scheduler.
\item $W_P$, the computation work of a P-core the parallel program.
  Here, $W_P$ plays has the same role as before, 
   but it explicitly excludes the scheduling work, which is no longer neglected.
\end{itemize}
Even though the 4 quantities above are hard to measure directly, we can use them
to help us to provide valuable interpretation to the curves from factored speedup plots.

We define the \E{scheduling work inflation}, written $F_P^{\F{sched}}$, as
the difference between the scheduling work performed by $P$ cores and that
performed by a single core.
Symmetrically, we define the \E{user work inflation}, written $F_P^{\F{user}}$, as
the difference between the user work performed by $P$ cores and that performed
by a single core.
Finally, we define the \E{work inflation} to be the sum of the user work
inflation and the scheduling work inflation.
$$\begin{array}{l@{\;\,}l@{\;\,}l}
F_P^{\F{sched}} & {\equiv}& S_P - S_1  \vspace{3pt}\\
F_P^{\F{user}} &{\equiv} &W_P - W_1  \vspace{3pt}\\
F_P &{\equiv}& F_P^{\F{user}} + F_P^{\F{sched}}
\end{array}$$

As we are going to establish next, the value of $F_P$, which denotes the total
work inflation, can be computed from the same four measures as before.
To prove it, we begin with two simple observations.

\begin{fact}[decomposition of $1$-core execution times]
The execution time of $1$-core run of the parallel program
decomposes as user work plus scheduling work.
\[
T_1 \Sc{=}  W_1 + S_1. 
\]
\end{fact}

\begin{fact}[decomposition of $P$-core execution times]
The execution time of a $P$-core run of the parallel program
decomposes as user work, plus scheduling work, plus idle time.
\[
P \cdot T_P \Sc{=}  W_P + S_P + I_P. 
\]
\end{fact}

Combining these two facts and the three definitions above
shows that the total work inflation can be computed using 
exactly the same formula as before (recall \fctref{formula-work-inflation}),
when we ignored all scheduling costs.
\begin{fact}[formula for work inflation, with scheduling costs]
\[
F_P \Sc{=}  P \cdot T_P - I_P - T_1. 
\]
\end{fact}

When taking into account scheduling costs, we continue using exactly the
same formulae for constructing factored speedup plots.
Only the interpretation of these plots needs to be refined slightly.
\begin{itemize}
\item The 1-core work ($T_1$) now includes the scheduling work at 1-core ($S_1$).
So, the maximal speedup curves includes not just the algorithmic overheads but also
the scheduling work at 1-core. 
\item The idle-time specific is, as before, based on
the maximal speedup, so it includes the scheduling work at 1-core.
\item The inflation-specific speedup also includes the scheduling work,
but also the scheduling work inflation. As we explain next,
the scheduling work inflation is typically negligible.
\end{itemize}

When using a work stealing scheduler, the scheduling work inflation only includes the
cost of performing load balancing and, if using concurrent deques, the possible increase
in the cost of accessing the deques due to concurrent accesses to the same cache lines.
For most practial applications, the number of steals is relatively small and 
the accesses to the deques are relatively cheap in front of the work performed by
the threads, so the scheduling work inflation ($F_P^{\F{sched}}$) is negligible. 
In such a case, the work inflation ($F_P$) can be considered equivalent to the user 
work inflation ($F_P^{\F{user}}$).

As a concluding remark, we observe that the scheduling work at 1-core ($S_1$) 
can be estimated by running the \E{sequential elision} of the parallel program.
This elision consists of a copy of the parallel code in which all parallelism
constructs are replaced with sequential constructs. For example, fork-join
operations are replaced with simple sequences. Let $T_{\F{elision}}$ denote
the execution time of the sequential elision. We can estimate $S_1$ by
considering the difference with the 1-core execution time of the parallel program.
In other words, $S_1 \Sc\approx T_1 - T_{\F{elision}}$.
When the sequential elision program is available, we can extend the factored 
speedup plot to report its execution time. To that end, we add an extra curve,
located above the maximal speedup curve, showing the points at height:
${(P\cdot T_s)}/{T_{\F{elision}}}$.
This additional curve is useful in particular to easily spot issues related
to granularity control, whereby the creation of too-small tasks imposes
significant scheduling overheads. When granularity control is performed
properly and an efficient scheduler is used, the new curve should collapse
with that of maximal speedups, reflecting the fact that scheduling costs ($S_1$)
are neglible.

\section{Case studies}

This section illustrates the application of our method on a multicore
machine with several different runs of a few benchmark programs. We
ported these programs from well-established benchmark suites, such as
the Cilk benchmarks and the Problem Based Benchmark Suite, to our
scheduling library. Although we selected only a few benchmark
programs, we emphasize that the methods we use are readily applicable
to any of the other benchmark programs in the respective suites and,
more generally, to any Cilk program.

\paragraph{Experimental setup.}
We conducted all the experiments described on our 40-core test machine.
The machine hosts 4 Intel E7-4870
chips running at 2.4GHz and has 1Tb of RAM. Each chip has 10 cores and shares
a 30Mb L3 cache. Each core has 256Kb of L2 cache
and 32Kb of L1 cache, and hosts 2 SMT threads, giving a total 
of 80 hardware threads, but to avoid complications with hyperthreading
we did not use more than 40 threads.
The system runs Ubuntu Linux (kernel
version 3.2.0-43-generic).
We also ran the same experiments on a 48-core AMD machine%
which features a deeper memory hierarchy,
and observed similar results.
All our programs are implemented in C++, compiled with GCC 4.8,
and rely on the scheduling library PASL,
which itself relies on a work stealing scheduler.
PASL provides two schedulers: one implemented with
concurrent deques (like in Cilk), and another one implemented using
private deques (see \cite{acarchra13}).  Both schedulers gave
similar results on the benchmarks described in the present paper.

\subsection{Case study 1: typical factored speedup plots}
\label{sec:case-study1}

\begin{figure*}[t]
\small
\begin{center}
\begin{tabular}{@{}c@{}c@{}c@{}c@{}c@{}}
      \hspace{-0em}
      \includegraphics[width=\smallwidth,height=\smallheight]{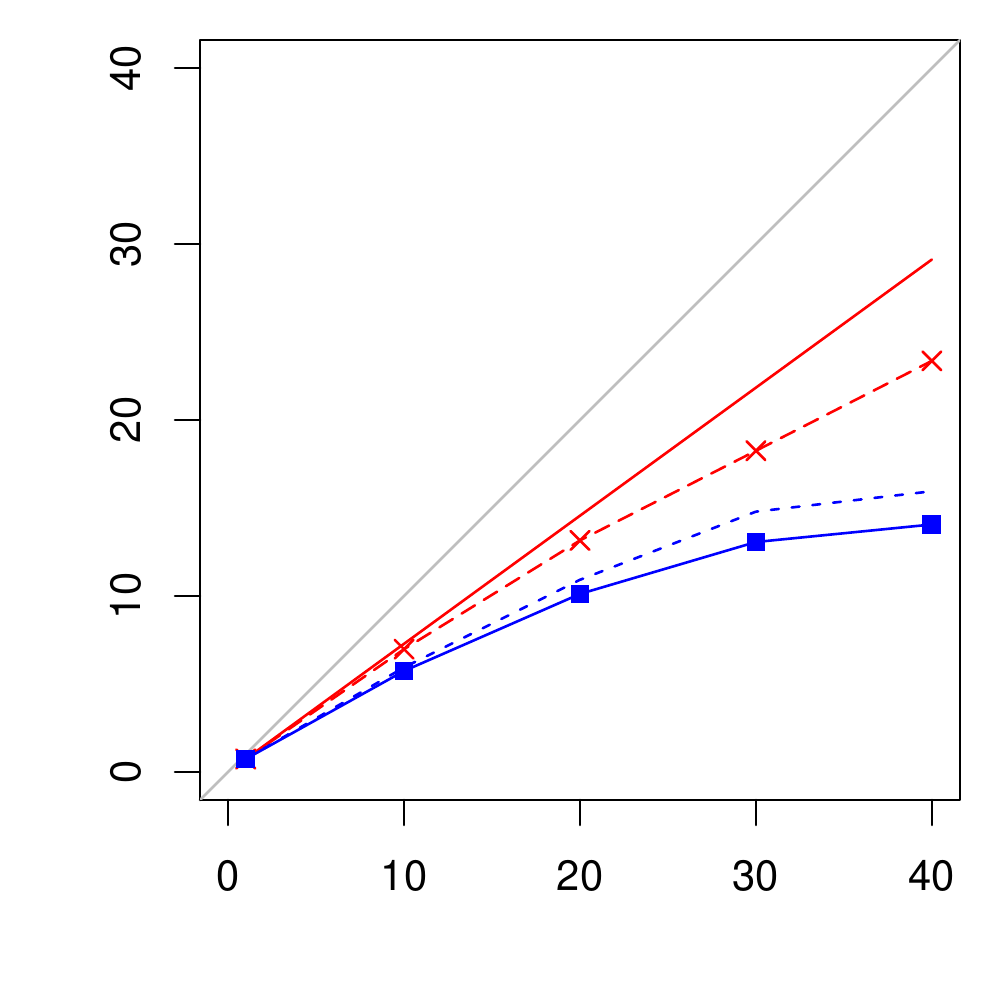}\figanalysisadjustheight        \vspace{-2em} &
      \hspace{-0em}
      \includegraphics[width=\smallwidth,height=\smallheight]{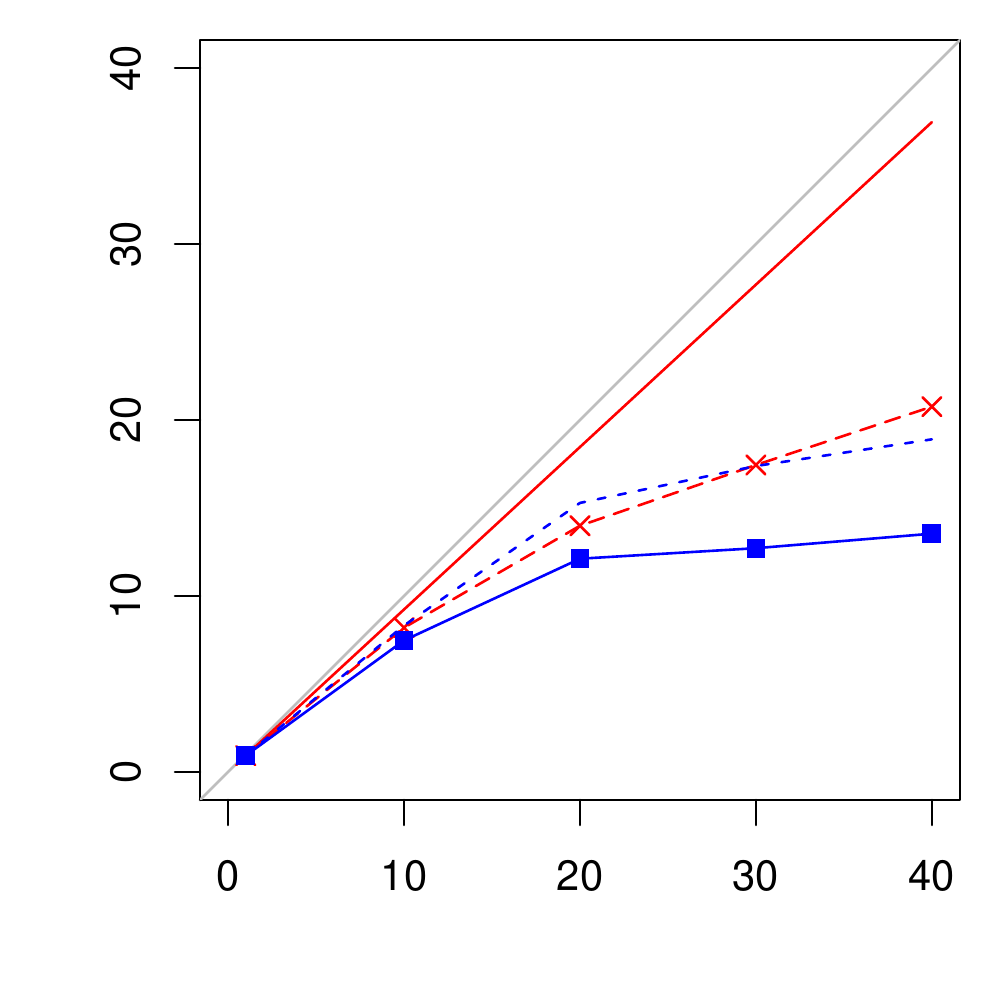}\figanalysisadjustheight &
      \hspace{-0em}
      \includegraphics[width=\smallwidth,height=\smallheight]{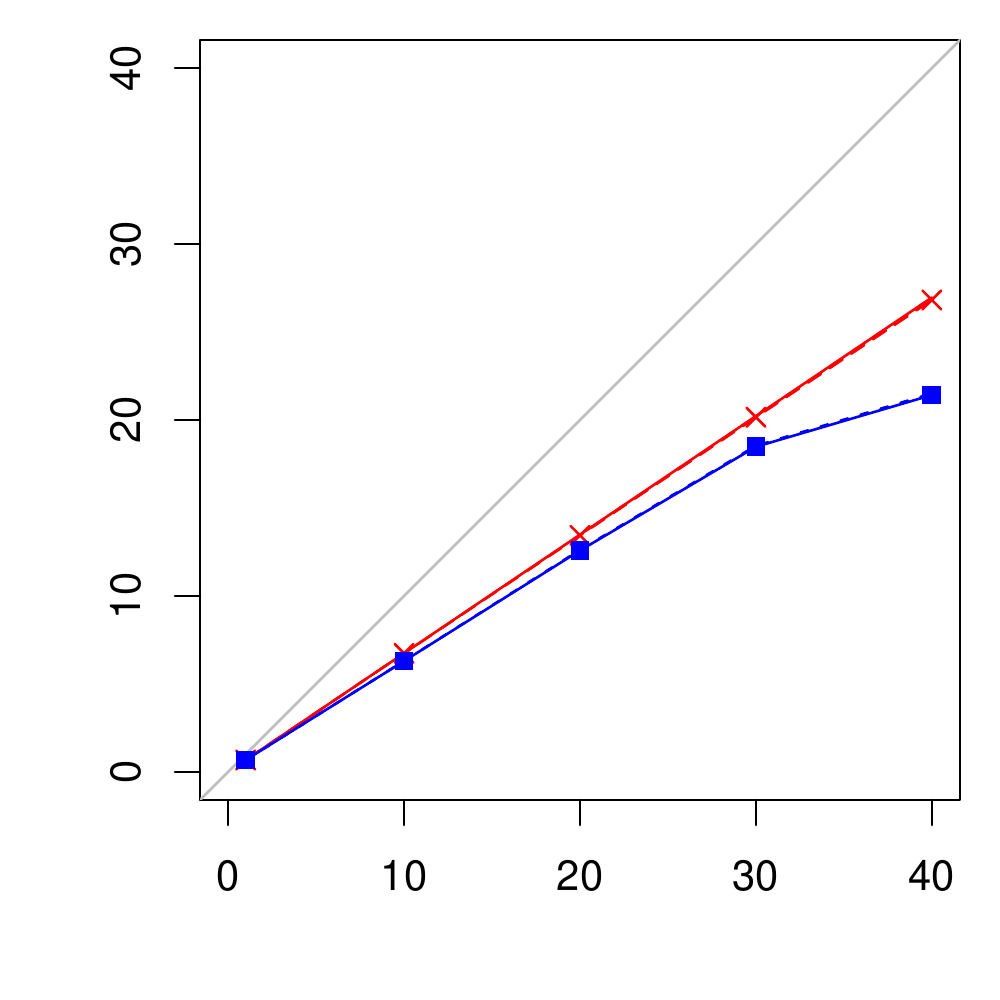}\figanalysisadjustheight &
\\
\quad(a) small cutoff, small array
&
\quad(b) very large cutoff, small array
&
\quad(c) small cutoff, large array
\end{tabular}
\end{center}%

\begin{center}
\begin{tabular}{@{}c@{}c@{}c@{}c@{}c@{}}
      \hspace{-0em}
      \includegraphics[width=\smallwidth,height=\smallheight]{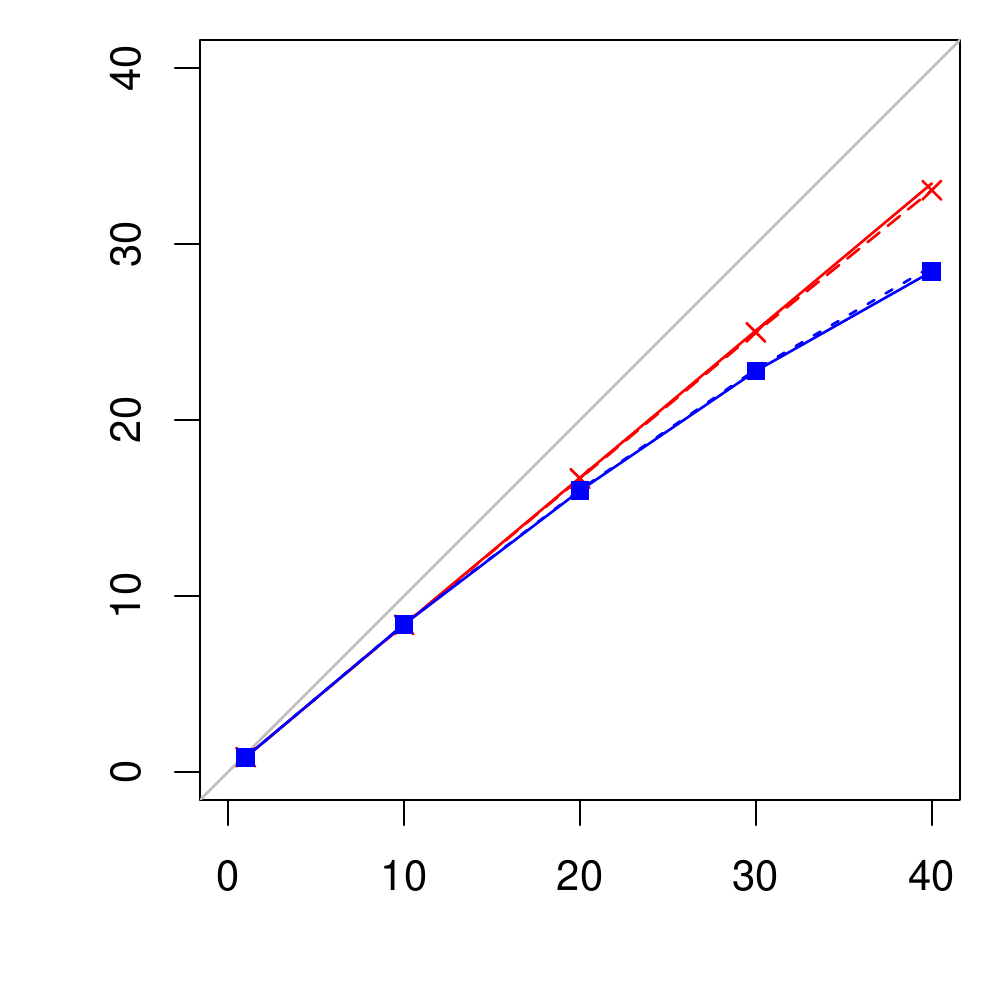}\figanalysisadjustheight \vspace{-2em} &
      \hspace{-0em}
      \includegraphics[width=\smallwidth,height=\smallheight]{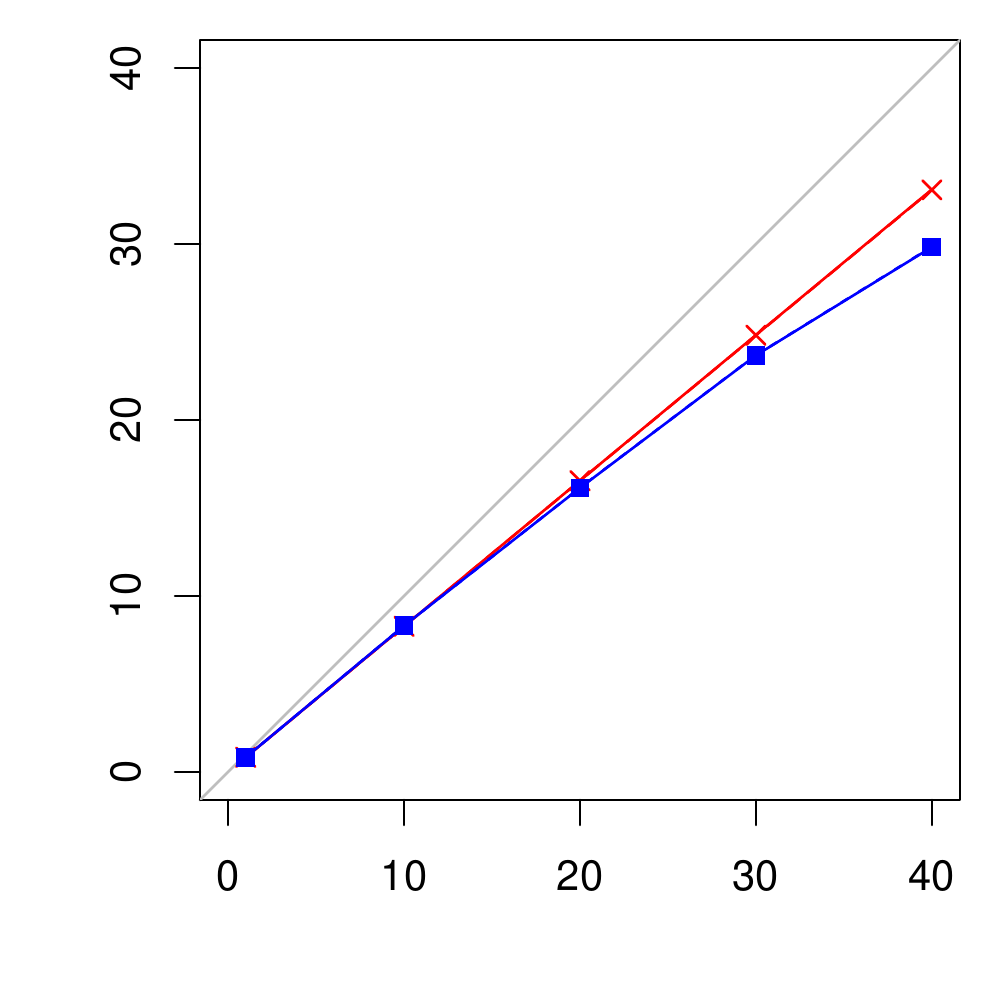}\figanalysisadjustheight &
      \hspace{-0em}
      \includegraphics[width=\smallwidth,height=\smallheight]{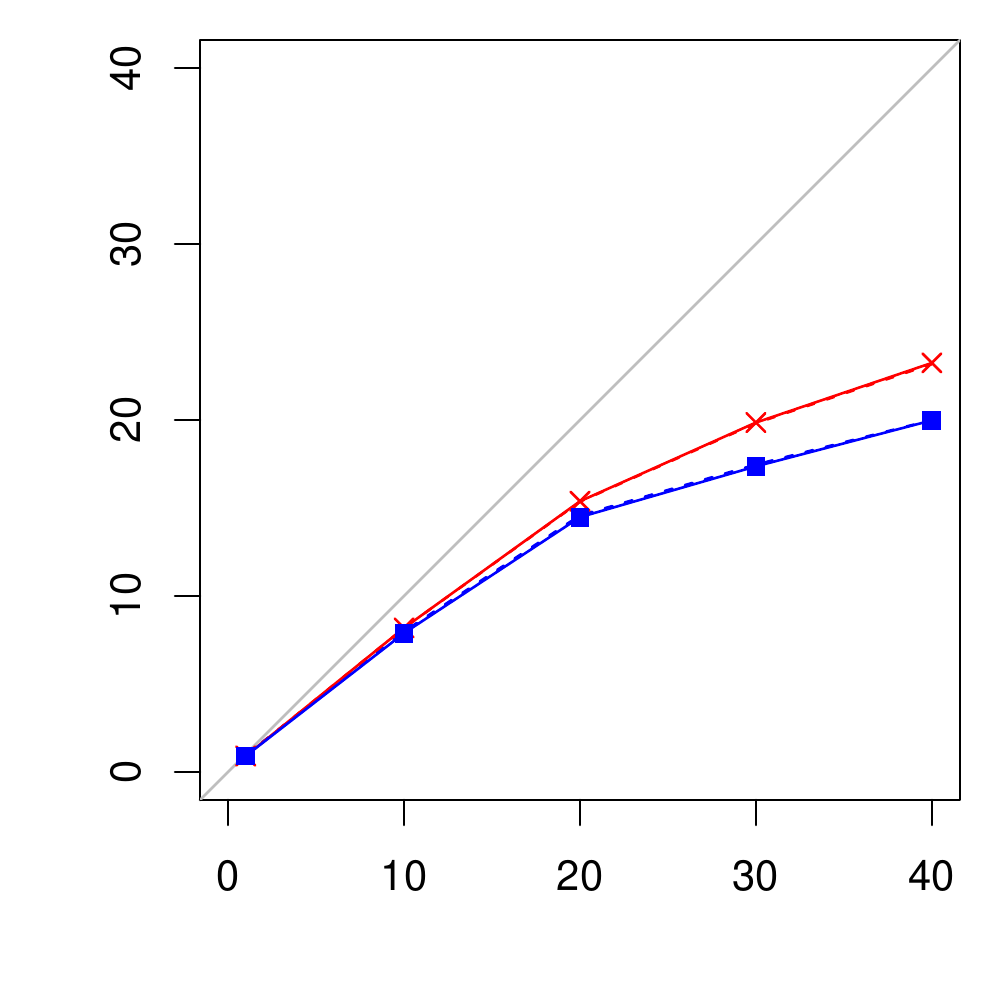}\figanalysisadjustheight &
\\
\quad(d) large cutoff, large array
&
\quad(e) large cutoff, huge array
&
\quad(f) cutoff in $O(\frac{1}{P})$, large array
\end{tabular}
\end{center}%

\caption{Case study 2: factored speedup curves for Cilksort
  benchmark.
  The straight, black curve represents the maximal speedup
  curve, the dotted, crossed one the idle-time specific curve, the
  dotted, blue one the inflation-specific curve, and the solid, blue
  one the actual speeedup curve.}
\label{fig:factored-speedups}
\end{figure*}

To illustrate the utility of factored speedup plots in practice, we
consider a classic benchmark program, namely Cilksort, and use
factored speedup plots to analyze its performance.  Cilksort sorts an
array of 32-bit integers, using a variant of merge-sort that relies on
a parallel merge operation, and relying on insertion-sort for sorting
sub-arrays of 20 elements or fewer. When the input is smaller than a
user specified {\em cutoff}, Cilksort reverts to sequential execution.
Sequentialized sorting uses the quicksort algorithm. Quicksort is also
used to measure the sequential baseline, used when computing speedup
values.  

In our experiment, we control the size of the input array,
and the cutoff, to determine how Cilksort behaves under different
settings.  (We use the same cutoff for both the sort phase and the
merge phase.)  
The goal of our experiments is to illustrate various typical 
type of factored speedup plots that one observe in practice.

In {\bf Chart (\ref{fig:factored-speedups}.a)}, 
we consider a small array, containing 200k items, and a
small cutoff, of 200 items.  This chart indicates that our
program suffers simultaneously from three problems: large parallel
work as indicated by the gap between the linear and the maximal speedup
curves, scarce parallelism as indicated by the gap between the maximal
and the idle-time-specific speedup curves, and work inflation as
indicated by the idle-time-specific and actual speedup curves.

In {\bf Chart (\ref{fig:factored-speedups}.b)}, we attempt to reduce 
the thread-creation overheads by increasing the cutoff size to 10k items.
The gap between the maximal speedup and the linear speedup closes indicating 
that we have successfully reduced parallel work.  Actual speedups, however, have
not improved---there are actually slightly worse---because parallelism
reduced as indicated by the increased gap between the maximal and the
idle-time-specific speedup curves. This suggests the cutoff is too
large for this input, pinpointing exactly the source of the problem.
Remark: the fact that the inflation-specific and idle-time-specific
speedups are at about the same height indicates that both work
inflation and idle time contribute to roughtly the same amount of
lost speedups.

In {\bf Chart (\ref{fig:factored-speedups}.c)}, we revert to the
smaller cutoff of  200 items and try instead to address the lack of
parallelism, by increasing the array size to 10 million items.
The overheads in Chart (\ref{fig:factored-speedups}.c)
are similar to those of Chart (\ref{fig:factored-speedups}.a),
which is expected since we used the same cutoff value.
The idle time has been reduced significantly, thanks to
the increase in the amount of parallelism available.
In this chart and the subsequent ones, the amount of idle time
is negligible, so idle-specific curves collapse onto 
maximal speedup curves, and inflation-specific curves collapse
onto actual speedup curves.

In {\bf Chart (\ref{fig:factored-speedups}.d)}, we target a
large array of 10m items and use a not-too-small cutoff of 1000 items.
The chart reports decent speedups (28.5x at 40 cores), and the trend
of the speedup curve suggests good scalability.

In {\bf Chart (\ref{fig:factored-speedups}.e)}, we increase
further the array size, up of 100m items, 
while keeping the same cutoff of 1000 items. 
The results are very similar to Chart (\ref{fig:factored-speedups}.d),
only with slightly better speedups (29.8x at 40 cores),
showing that beyond a certain point, creating more parallelism
no longer reduces the idle time.
In fact, from the position of the overhead curve, which reaches
33.1x at 40 cores, we can deduce that, no matter the array size, 
it is highly unlikely to ever exceed a speedup of 33.1x on our test machine.

With {\bf Chart (\ref{fig:factored-speedups}.f)}, 
we complete our case study with a last experiment which 
aims at illustrating a situation
where the amount of work varies with the number of processors. 
To that end, we provide to Cilksort a cutoff inversely proportional 
to the number of processors. Note that adapting the number of subtasks generated
to the number of processors is a classic technique, 
used for example in Cilk's compilation of for-loops.

For this last experiment, we consider an array of size 10m and a cutoff of $8000/P$.
As the value of $T_1$ actually depends on $P$, we write it $T_1^P$.
To obtain the values of $T_1^P$, we perform, for each value of $P$,
a single-processor run using the cutoff value $8000/P$.
The results, shown in Chart (\ref{fig:factored-speedups}.f),
indicate that the idle time is negligible, that the memory effects
are very limited, and that overheads are responsible for most of the lost speedups.
Furthermore, on the chart we are able to observe the curvature of the overhead 
curve. The fact that the overhead curve is not a straight line
but instead bends downwards indicates that the amount of overhead
increases with the number of processors.



In summary, by looking at the curvature of the curves
and the space between the curves of factored speedup plots,
we are able to visualize, all at once, 
the relative contribution to the loss in speedups of each of the 
three possible sources of slowdown identified by our model,
and also to visualize the trends of these contributions 
as the number of processors vary.

\subsection{Case study 2: effect of NUMA allocation policies}
\label{sec:case-study2}

We now describe how our factored speedup plot can be used to diagnose
memory bottlenecks. For this study, we consider the Maximal
Independent Set benchmark from the Problem Based Benchmark Suite. The
maximal independent set problem is the following: given a connected
undirected graph $G = (V,E)$ and find a subset of the vertices $U
\subset V$ such that no vertices in $U$ are neighbors in $G$ and all
vertices in $V \, \backslash \, U$ have a neighbor in $U$. For input
to the benchark, we used the 2-d grid with $140$m vertices. For the
baseline measurement, we use the sequential solution that is provided
by the Problem Based Benchmark Suite. The performance issue we
consider came to our attention when we ported the program from the
Cilk Plus dialect of C++ to be compatible with our native C++
scheduling library, namely PASL~\cite{acarchra13}.

The plots in \figref{case-study2} show two factored speedup plots
representing two different NUMA configurations of the same
application. The runs of plot (a) and (b) use the default and the
interleaved NUMA configurations respectively. We describe the meanings
of the two configurations after first considering the results we
observe from the default configuration. In plot (a), we notice that
the actual speedup curve starts to flatten by ten processors and
completely flattens by twenty. The flattening of this curve happens
even though there is clearly no lack of parallelism: we know there is
sufficient parallelism because the idle-time specific curve hugs the
maximal curve. The inflation-specific curve shows that the most
significant factor harming scalability is work inflation.

Knowledge of our machine led us to the next step, that is, to
conjecture that significant work inflation is imposed by effects
relating to non-uniform memory access (a.k.a, NUMA). NUMA implies that
memory-access time depends on the memory location relative to which
processor makes the access. Our benchmarking machine has four banks of
RAM, with one bank assigned to each physical chip in the machine. Each
bank of RAM is close to the ten cores on its corresponding chip and is
far from all the other cores. We suspected NUMA effects because
scaling drops significantly only when the number of cores exceeds
ten. This point is the point at which at least some of the cores have
to make remote accesses to access main memory.

We investigated the NUMA policies that are supported by our machine
and found that there are two of interest. In the default
configuration, namely the local or ``first-touch'' configuration, a
page in virtual memory is assigned a page in physical memory when the
page is first accessed. The page is assigned in physical memory to
memory bank of the core that makes the first access. The other
configuration of interest is the interleaved configuration, in which
pages are assigned to memory banks in round-robin fashion.  Although
the interleaved configuration increases cross-bank traffic relative to
the first-touch configuration, the interleaved configuration reduces
the chance of a bottleneck situation, in which much more memory
traffic goes through a few banks of RAM than through other banks.

Suspicious of such a bottleneck, we tried the interleaved NUMA
configuration. The actual speedup we get from this configuration is
shown in Figure \ref{fig:case-study2}(b).  Note that we can compare
the spedups of the two plots because all of the speedup curves use the
same baseline. The speedup achieved by the configuration is much
better than before, suggesting that, in the default configuration,
there was significant imbalance of NUMA assignments leading to
contention at the memory bus.

With these plots side by side, we can see additional patterns in the
respective curves. Observe that, even though it shows relatively poor
actual speedup, the first plot shows better maximal speedup. The
reason is that the single-processor run of the program runs faster
with the local than with the interleaved NUMA configuration. In other
words, the same NUMA configuration that harms the performance of the
sequential run helps the performance of the parallel run. Moreover,
this particular improvement comes into effect when the number of cores
exceeds ten, because the effect is a NUMA effect.

To summarize, while the factored speedups provided all the information
we needed to diagnose the NUMA issue, the curves gave us a clear
picture of where to start looking. In particular, the fact that the
curve flattens between ten and twenty processors gave us a strong hint
that the issue is NUMA related.

\begin{figure}[t]
\small
\hspace{-2.6em}
\begin{tabular}{@{}c@{}c@{}c@{}}
      \hspace{-0em}
      \includegraphics[width=1.9in,height=1.9in]{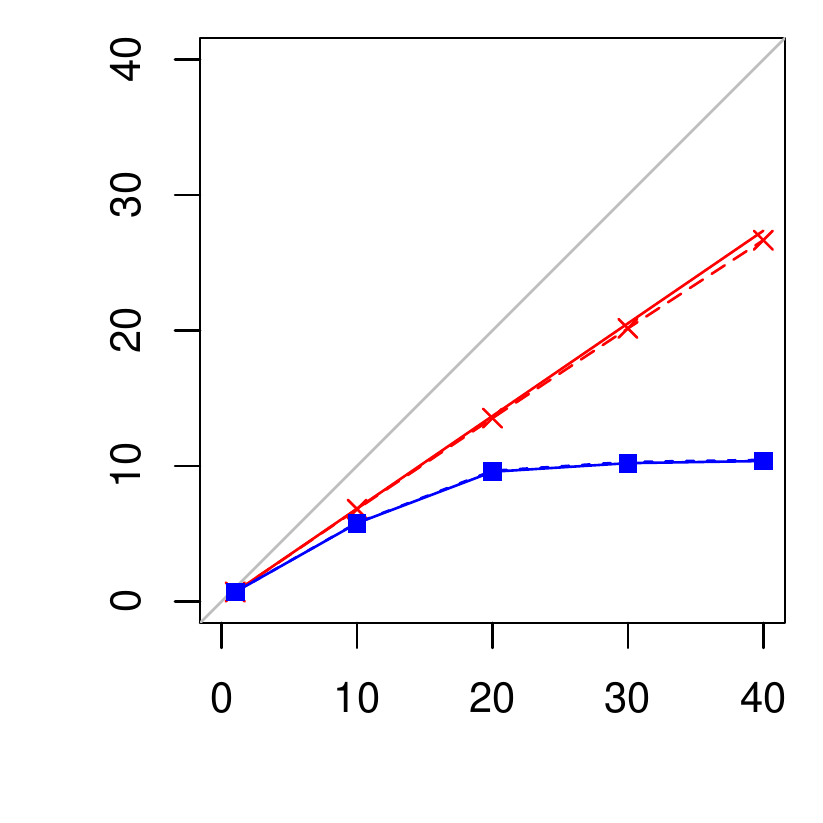}\figanalysisadjustheight \vspace{-2em} &
      \hspace{-1.5em}
      \includegraphics[width=1.9in,height=1.9in]{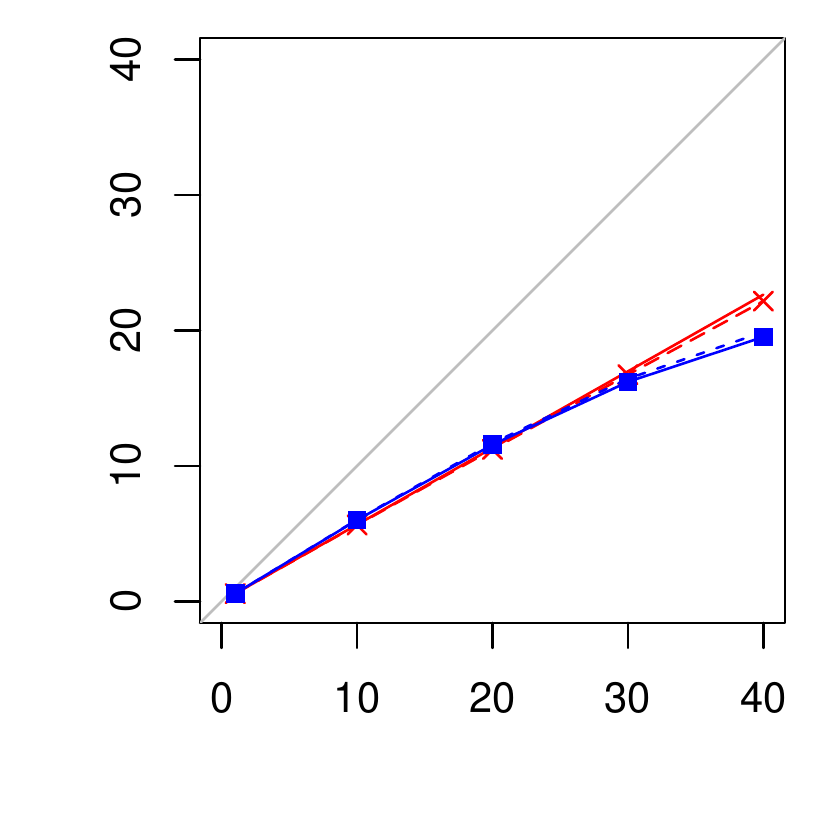}\figanalysisadjustheight &
\\
\quad(a) default
&
\quad(b) interleaved
\end{tabular}
\caption{Factored speedup curves for Maximal Independent Set
  benchmark.
The straight, black curve represents the maximal speedup
  curve, the dotted, crossed one the idle-time specific curve, the
  dotted, blue one the inflation-specific curve, and the solid, blue
  one the actual speeedup curve.}
\label{fig:case-study2}
\end{figure}

\section{Sources of Work Inflation}
\label{sec:sources}


In this section, we present what we believe to be two particularly
striking and subtle causes of work inflation.  To simplify their
presentation, we distill the causes of the work inflation in
simplified benchmarks.  Our measurements show that work inflation can
affect speedups by nearly a factor two. 
%
%
In particular, we show that the speedups achieved may greatly vary 
with the size of the input data considered, and that 
they may greatly vary with the degree of optimizations that applies to 
pieces of code involved both in the baseline program and in the parallel program.
In such circumstances, a higher degree of optimizations (which leads to reduced 
absolute execution time) may lead to smaller speedup values.

\paragraph{The benchmark.} To illustrate work inflation, we use a
simple array microbenchmark, which is controlled by three parameters:
array size $M$, a computation load $L$, a gap size $G$,
and a number of repetition $R$.
Given a set of values, the benchmark starts by
allocating $M$ 
cells each of which contain a single 64-bit integer.
The program then processes every cell of the array once,
and repeats this entire process $R$ times.
To process a cell $c$, the benchmark performs
$L$ integer additions using the value at $c$ and writes the resulting value back
into $c$. We implement the parallel for-loop by dividing the total
range until a sufficiently small range of 1000 items, \aremark{check it was 1000...}
which are then processed sequentially.

When the gap size $G$ is equal to $1$, each thread processes
a group of 1000 consecutive array items sequentially. 
When the gap size $G$ is more than $1$,
threads still process groups of 1000 items, but acting over
items spaced out by $G$ cells, in such a way that, ultimately,
each array cell gets processed exactly once.
To be precise, the $i$-th cell processed is that at index 
``$(iG + \lfloor \frac{iG}{M} \rfloor) \F{mod} M$'' in the array.
By considering values of $G$ greater than $1$, for example $32$,
we are able to greatly increase the number of cache misses.

\paragraph{Input size and work inflation.} 

Our first experiments illustrate an interesting relationship between
input data size speedups. On the one hand, it is well-known that, with small
inputs, parallel programs may not generate sufficient parallelism to
result in good speedups. 
On the other hand, large inputs that do not fit in the L3 cache
lead to numerous cache misses, and they are typically associated with 
important levels of work inflation because the main memory becomes the bottleneck.
As we show, however, there can be a range of input instances large enough 
to generate abundant parallelism, and nevertheless small enough to 
avoid significant work inflation. With such input instances, one is able to
measure speedup values much greater than speedups that could be
achieved when scaling to a larger number of cores or to larger input instances.

\newcommand{\sizeexpsize}{2.5in}

\begin{figure}[t]
\small
\begin{center}
\includegraphics[width=\sizeexpsize,height=\sizeexpsize]{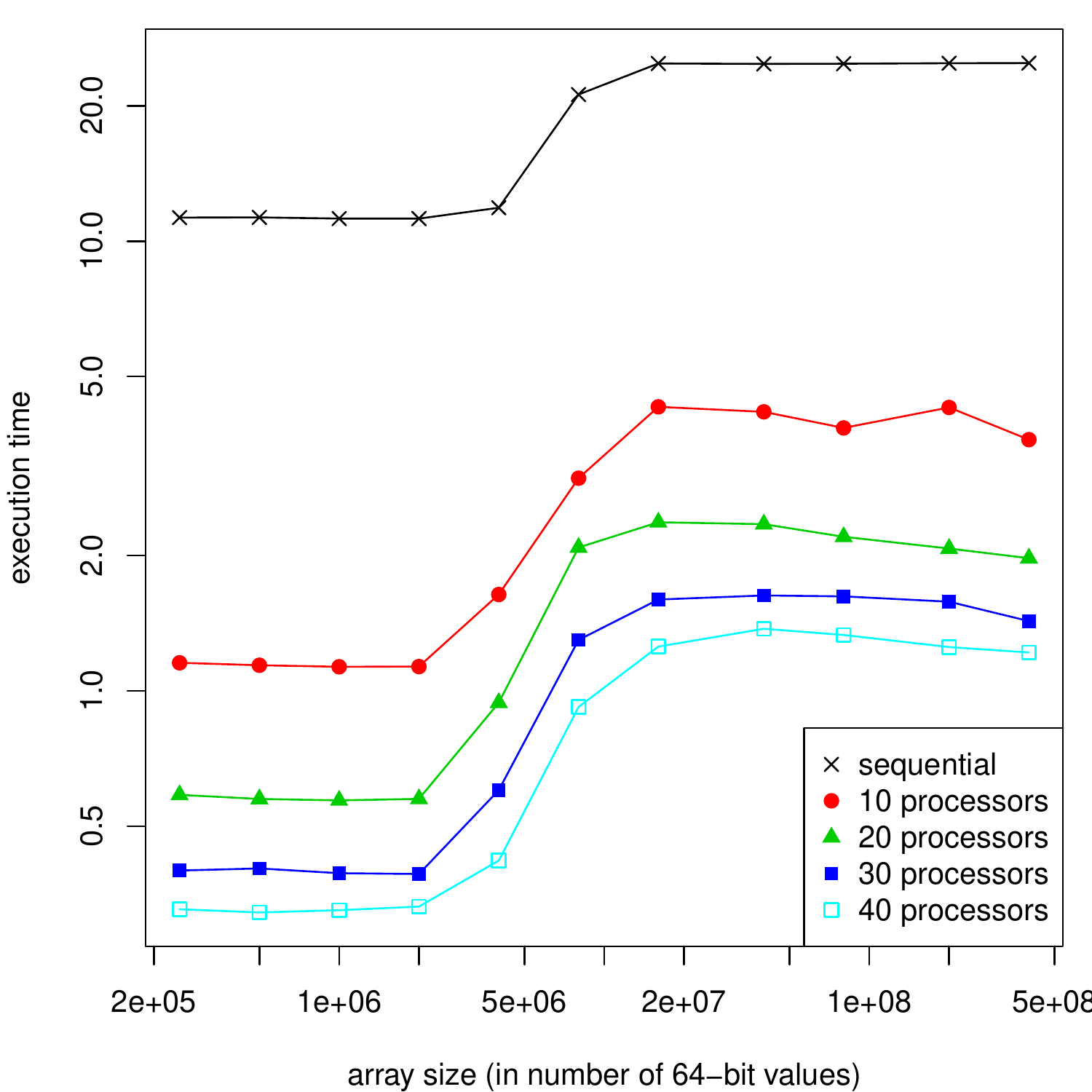} 
\end{center}%
\begin{center}
\includegraphics[width=\sizeexpsize,height=\sizeexpsize]{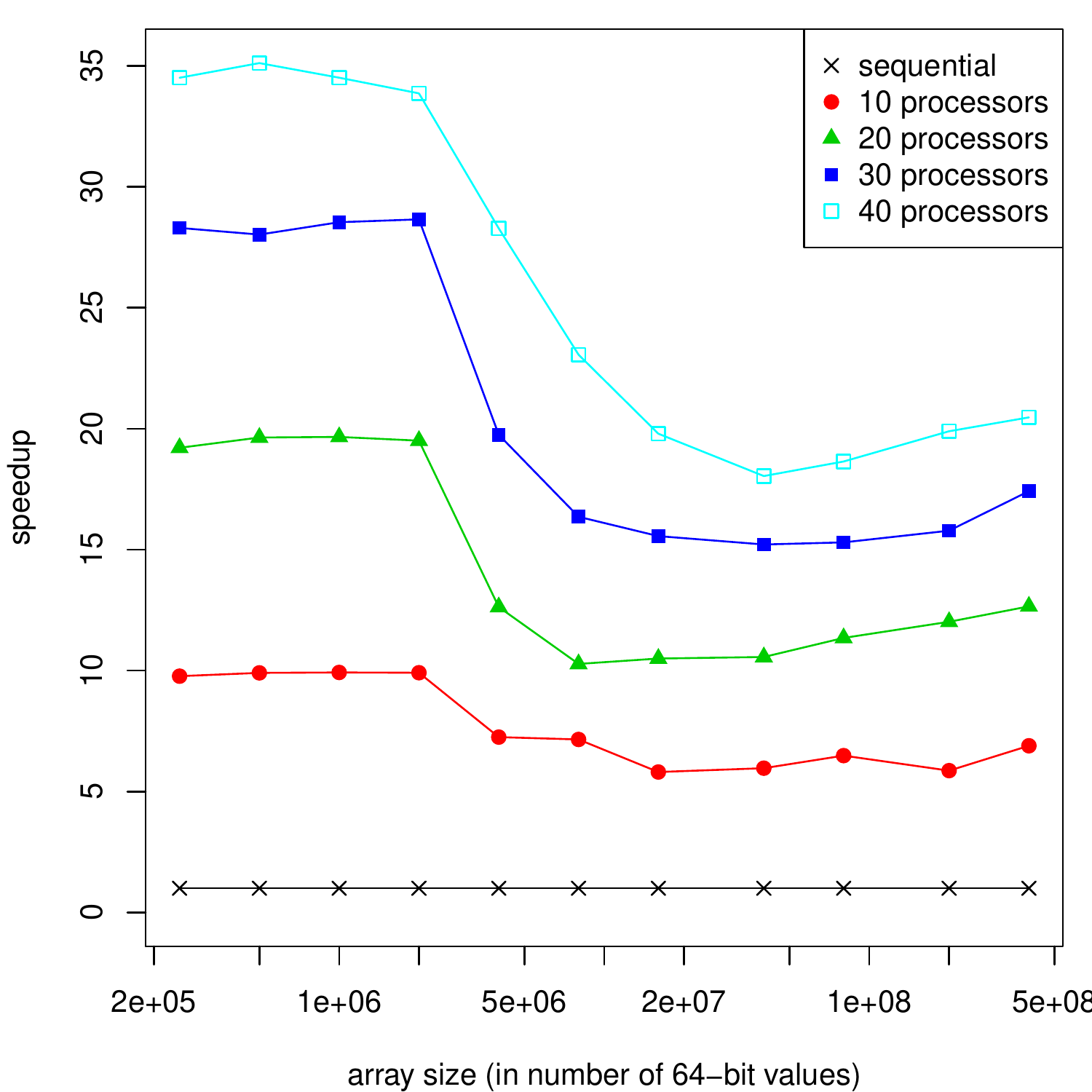} 
\end{center}%
\caption{Runtime (top) and speedup (bottom) versus array size illustrated.}
\label{fig:exp-size}
\end{figure}

\figref{exp-size} illustrates the runtime and speedup for our
microbenchmark with different array sizes $M$ and different numbers of
processors.  In these experiments, we set the gap size to be $G=32$, and
set the repeat count $R$ to be $\frac{4\cdot 10^8}{M}$ so that the
total number of operations (a measure of the complexity of the
benchmark) remains the same for all input sizes (i.e., $4\cdot 10^8$). 
The runtime curve (\figref{exp-size}, top) shows that compared with
small input sizes, a sequential run, the topmost curve, is 2.2 times
slower for inputs larger than $16 \cdot 10^6$ ---increasing from
11.3 seconds to 24.9 seconds. This outcome is
expected, because the 30MB L3 cache of this processor approximately $4
\cdot 10^6$ (64-bit) integers.  What is interesting is that the slowdown is
amplified in parallel runs.  For example with $30$ cores, larger arrays
are $3.7$ times slower compared with the smaller ---increasing from 0.37 seconds
to up to 1.37 seconds.
While it is generally known that higher number of cache misses
slow down a program execution, what is interesting here is that 
this slow down affects performance differently at different sizes.
This behavior is likely due to the saturation of the memory bus 
at high parallel loads.

The fact that, when increasing the array size, parallel runs are slowed 
down more than sequential runs indicates that the work inflation
increases with the array size. A direct consequence is that, as shown by
the curve at the bottom in \figref{exp-size}, speedups can decrease
significantly when operating on larger arrays.  For example, 
with 40 cores, the speedup for small array is close to 35x, but 
with larger arrays it drops below 20x.

In summary, while with small inputs, the benchmark achives nearly
perfect speedups, at large input sizes, the speedups decrease
significantly.  This suggests that work inflation can be significant
and it should be accounted for by considering a range of input sizes,
not just those input sizes that provide sufficient parallelism. 

\paragraph{Work inflation and optimization.}

Since speedups are calculated with respect to a baseline sequential
program by calculating the ratio of the runtime of the
sequential baseline to the runtime of the parallel code, it might be
concluded that optimizing both programs to the same degree would
suffice to perform a fair evaluation.  In fact, the parallel
code is often written by using the pieces of the sequential code, as this
is often the easy and the natural thing to do. 
%
As we show next, speedups can be highly sensitive
optimizations, not just because optimizations can improve the baseline
performance---which is generally known and understood---but also
because optimizations can impact serial and parallel code in different
ways, by leading to different amounts of work inflation.

\begin{figure}[t]
\small
\begin{center}
\includegraphics[width=\sizeexpsize,height=\sizeexpsize]{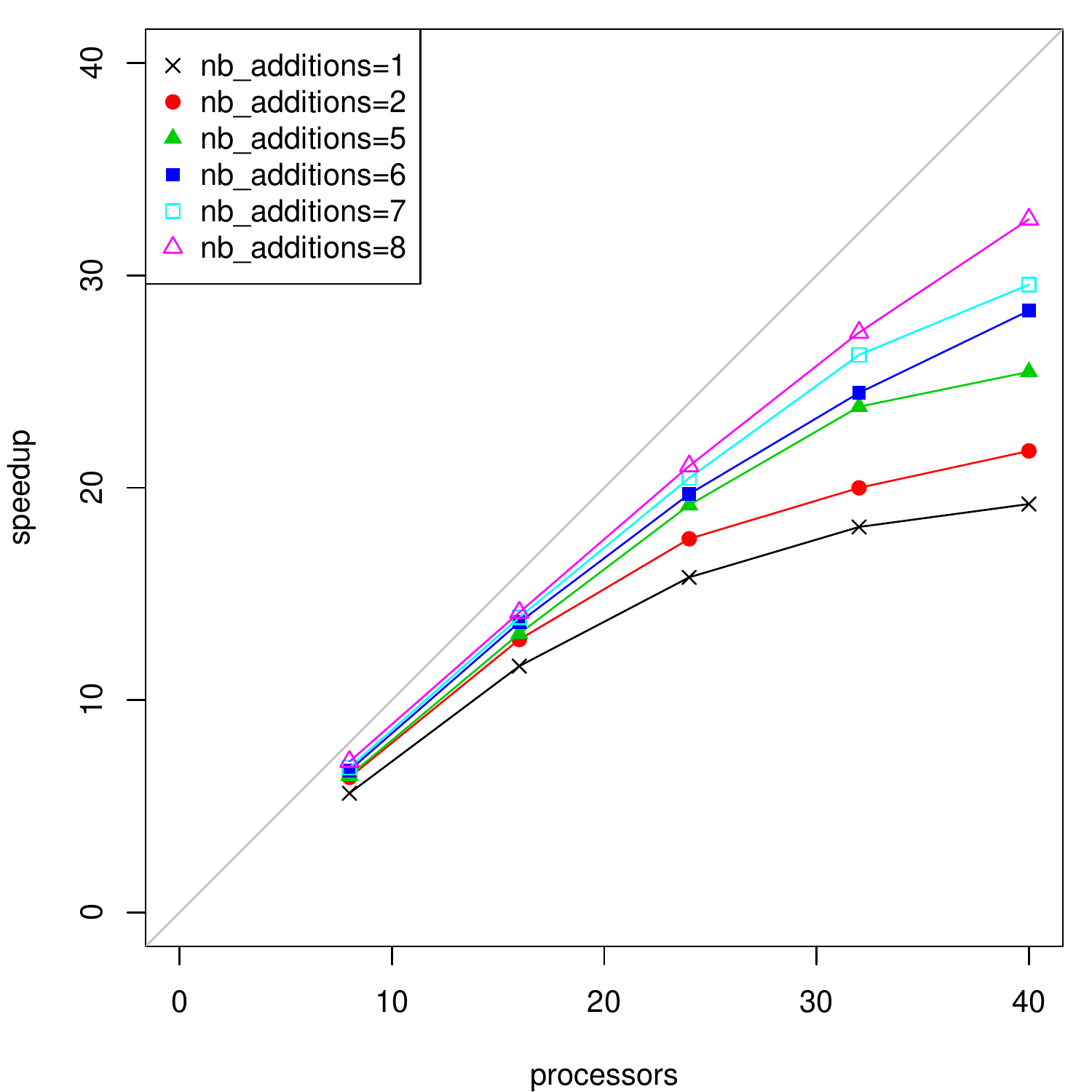} 
\end{center}%
\caption{Impact of the computational load (number of additions
  performed between memory operations with gap size $1$.
}
\label{fig:exp-artificial-gap}
\end{figure}

To demonstrate the effect of optimization on work inflation, we
consider our simple microbenchmark and run it with $M = 600 \cdot
10^6$ (that is, a 4.8Gb array), $R = 1$, and different values of computational
load $L$ ranging from $1$ to $8$.  Recall that the microbenchmark
performs $L$ additions after reading a cell and writes back the
computed value to the memory.  The differing values of $L$ suggest
what can happen with highly optimized code $L = 1$ and poorly
optimized code $L = 8$. 

The plot \figref{exp-artificial-gap} shows the curves for %
different values of $L$ that we consider. The measurements show that
the more the additions, the better the speedups.
The implication is that 
additional work due to more additions creates relatively less work
inflation.  
This implication is likely true because in parallel runs, all
computation becomes memory bound, waiting for the memory operations to
complete, during which time, cores can perform the addition operations
(which commute) locally, without having the value of the cell being
updated until it finally arrives.  This property implies that the
addition operations are parallelized by the hardware to overlap with
the memory operations, reducing the relative significance of work
inflation.  We tested this hypothesis in two ways.  First, we changed
the addition operations to operations to commute with the
reads; this change reduced the relative work inflation, ultimately improving
the speedups.  
Second, we ran the benchmark with larger values of $L$,
thereby increasing the memory latency for the sequential run, and thereby
decreasing the relative work inflation. 
 
In summary, when memory operations become a bottleneck in the parallel
run, increased computational load due to non-agressive optimization
can artifically increase speedup by reducing relative work inflation.
It is therefore not sufficient to optimize the sequential baseline and
the parallel code to the same degree.  The baseline as well as the
parallel code should be highly optimized in order to make sure that
the effects of work inflation are not masked.

\vspace{-4pt}
\section{Related work}

\paragraph{Prediction of parallel speedup.}

Cilkview~\cite{HeLeiserson10}, Intel Parallel
Advisor~\cite{IntelParallelAdvisor}, Intel Parallel
Amplifier~\cite{IntelParallelAmplifier}, and
Kismet~\cite{JeonGarciaLouie11} are software tools whose purpose is to
profile and to analyze the potential scalability of programs on an
arbitrary number of cores.  Cilkview, Intel Parallel Advisor, and
Intel Parallel Amplifier rely on user-supplied annotations, whereas
Kismet tries to automatically detect parallelism in the application.
Our method focuses instead on identifying the causes of suboptimal speedup of a
given parallel program on a given machine with a fixed set of cores.

\paragraph{Modeling parallel performance.}

Our techniques and those used for Cilkview share a common basis in the
DAG model of computation.  However, we use the DAG model in different
ways to achieve different goals.  On the one hand, the Cilkview
profiler measures the work and span during the instrumented run of a
parallel application on a single processor.  The Cilkview analyzer
predicts from the work and span the upper and lower bounds on the
speedup curves that can be achieved by the application on an arbitrary
number of processors.  On the other hand, based on a mix of sequential
and parallel runs, our analyzer plots, next to the actual speedup
curve, a synthetic speedup curve that projects the amount of speedup
lost due to idle time and parallelism overheads, allowing to visualize
the amount of speedups lost due to memory effects.


In Cilkview, work and span are measured by number of instructions
issued by the program, as opposed to wall-clock time.  By considering
instruction counts, the scalability prediction of Cilkview is
completely oblivious to memory effects that could substantially harm
scalability.  Our work, although it is limited in that it considers
only typical execution paths as opposed to worst-case execution paths,
is able to deduce the amount of memory effects that impact the
parallel runs.

Cilkview, being based on the work-span model, tries to evaluate the
span.  To that end, it considers a ``burdened-dag model'', where the
weight of fork nodes is burdened with an estimate of the cost of
thread migration. The span measured in this burdened DAG gives a
worst-case estimation of the span. In our work, we do not try to
measure the span at all. Instead, we rely on the measure of the actual
idle time, as explained in \sref{sec:compare}.  Cilkview may
nevertheless provide a complementary role in helping to estimate
worst-case bounds on the idle time.

\paragraph{Identifying sequential bottlenecks in big programs.}

The HPCToolkit~\cite{TallentMellorCrummey07, TallentMellorCrummey09} is a
software tool for profiling big parallel software that consists of
many functions.
HPCToolkit reports, on a per-function basis, estimated values of
parallel idle time and parallelism overheads.
Kremlin~\cite{GarciaJeonDonghwan11} is another software tool whose
purpose is to help guide the parallelization of large preexisting
sequential programs.
Kremlin, like HPCToolkit, 
focuses on the
question: \textit{what parts of the program are most profitable to
  parallelize?}
As such, the primary focus of these tools is to assign blame to
pieces of code that are imposing bottlenecks to parallelization.

In contrast, our focus is to analyze the performance of algorithms
individually rather than to try to analyze the relative performance of
multiple algorithms in the same program.
Put another way, our focus concerns the stage after the programmer has
identified a bottleneck code.
At this point, the goal is to isolate the code and benchmark it 
independently 
to try and improve its scalability.

Often, blame-assigning tools, such as HPCToolkit and Kremlin, neglect
to report in a synthetic way complementary pieces of information that 
would be helpful for understanding causes of poor speedup.
Our factored speedup plots show a global view of the actual parallel
performance of the optimized, production-ready code.
In addition to providing a synthetic view of the data, our factored speedup
plots show the speedup trends as the number of processors vary.
The trends are useful, among other things, for extrapolating the
ability of an algorithm to scale up to larger number of cores.

\paragraph{Profiling techniques.}

The aforementioned profilers, as well as other related 
ones~\cite{Reed93scalableperformance,MohrAllenShendeWolf02,
MooreWolfDongarraShendeAllenMohr05}, collect rich profiling
data from instrumented runs of an application.
Although sometimes useful, rich profiling data is not necessarily
the best approach.
Problematically, the instrumentation itself may affect the
performance of the application being profiled.
On the contrary, our approach relies on practically zero-overhead
instrumentation and as such can be applied to production-ready
user code.

In our approach, the required instrumentation consists of measurement
of run time of the sequential baseline program, single-processor run
time of the parallel program, run times of the parallel program on
different subsets of the available processors, and total parallel idle
time for each parallel run.
All of these metrics are trivial to measure and can be readily
measured in almost any platform.
Many other profilers require substantial implementation effort in
the form of compiler support or binary instrumentation. 

To summarize, while we acknowledge the interest of 
full-program analysis and of rich instrumentation, we have found
that our approach, despite being very lightweight, is able to report
a large amount of useful information helping to analyse the scalability 
issues affecting a particular parallel algorithm.

\vspace{-6pt}
\section{Conclusion}

On modern hardware, the impact of memory effects on the performance 
of parallel program is too important to be neglected. 
While these effects have shown difficult to model accurately,
developers of parallel programs could greatly benefit of tools
for analysing the relative impact of memory effects. 
In this paper, we have presented a simple model for the analysis of parallel
computations. Our model is tailored for the analysis of experimental 
performance results, and it aims an analysing samples of executions.
In that respect, it contrasts with the traditional work-span model,
which provides a theory for computing bounds for worst-case executions.

Our model is based on the simple observation that, by sampling the
execution time of single-processor runs and measuring idle time in
parallel runs, we are able to deduce the amount of memory effects.
Moreover, we have shown how to plot charts for visualizing
the amount of speedups lost due to overheads, that lost due to idle time, 
and that lost due to memory effects. These charts allow to visualize
not only the relative contribution of each source of slowdown, 
but also their trend as the number of processors grow. 
Although we have not seen such charts appear previously in the literature,
they are, in our experience, helpful for the day-to-day
development of parallel algorithms.



\small
\vspace{-5pt}
\bibliographystyle{plain}
\bibliography{main}

\end{document}

